\documentclass[12pt]{article}

 \usepackage{amsmath,amsthm,epsfig,amsfonts,amssymb}
 \usepackage[latin1]{inputenc}
 \usepackage{amstext}
 \usepackage{color}
 \usepackage{float}
 \usepackage{natbib}
 \usepackage{appendix}
 \usepackage{caption}
\usepackage{subcaption}
\newtheorem{proposition}{Proposition}

\theoremstyle{definition}
\newtheorem{definition}{Definition}[section]

\begin{document}
\begin{titlepage}
\begin{center}
{\large \textbf{Adaptive Financial Fraud Detection in Imbalanced Data with Time-Varying Poisson Processes}}\\
\vspace{2.5 cm}

R\'egis Houssou \protect\footnote{R\'egis Houssou, Research Data Scientist, HEIG-VD. E-mail: regis.houssou@heig-vd.ch} 
J\'er\^ome Bovay \protect\footnote{J\'er\^ome Bovay, Data Scientist, NetGuardians. E-mail: bovay@netguardians.ch}
Stephan Robert \protect\footnote{Stephan Robert, Professor, HEIG-VD, E-mail:Stephan.Robert@heig-vd.ch}\\
\vspace{0.5cm}
{\today}

\begin{abstract}
This paper discusses financial fraud detection in imbalanced dataset using homogeneous and non-homogeneous Poisson processes. The probability of predicting fraud on the financial transaction is derived. Applying our methodology to the financial dataset shows a better predicting power than a baseline approach, especially in the case of higher imbalanced data.\\
\vspace{1cm}
\end{abstract}

\vspace{6cm}
\end{center}
\end{titlepage}

\section{Introduction}
 Financial fraud is growing exponentially, especially because of the large sums involved. McAfee estimates in 2018 that cybercrime, of which financial fraud is a factor, costs the world about US\$ 600 billion, or 0.8\% of global GDP. According to McKinsey, global losses due to card fraud could reach nearly US\$ 44 billion by 2025. In addition to the direct cost of fraud, companies also suffer from lost sales when real transactions are denied by the companies. McKinsey estimates that false positives account for up to 25\% of transactions denied by online retailers, see \cite{Dyz}. However, as a first step, banks and financial institutions have approached the detection of fraud using manual procedures or rule-based solutions, which have yielded good results, but these methods currently have limitations. The rule-based approach means that a complex set of requirements for suspicious transaction reporting must be defined and reviewed manually. While this may be effective in detecting anomalies consistent with known patterns, it does not detect frauds that follow new or unknown patterns. The increasing complexity of digital attacks and the creativity of cyber-attackers make these conventional detection methods less effective and quickly obsolete. More sophisticated techniques must be developed, including automatic learning algorithms, and evolve the detection of fraud towards methods using adaptive rules to tighten the mesh of the network.\\
 
The machine learning models work with many parameters and are much more efficient at finding subtle correlations in the data, which can be masked by an expert system or by human criticism, \cite{Dyz}. The large volume of transactional data and client data readily available in the financial services industry makes it an ideal tool for the application of complex machine learning algorithms. In addition to learning from known models, machine learning can go further and learn new models without human operation. This allows models to adapt over time to discover previously unknown patterns or to identify new tactics that can be used by fraudsters. In fact, the development of conventional machine learning algorithms has led them to solve some specific problems, one of the most important features of which is that the distribution of data is generally balanced, unlike financial fraud, which is not balanced. Most standard classifiers such as decision trees and neural networks assume that learning samples are evenly distributed among different classes. However, in many real-world applications, the ratio of the minority class is very small( 1:100, 1:1000 or can be exceeded at 1:10000).  Due to the lack of data, few samples of the minority learning class tend to be falsely detected by the classifiers and the decision limit is therefore far from correct. Numerous research works in machine learning has been proposed to solve the problem of data imbalance; \cite{Garcia}, \cite{Galar}, \cite{Kraw}, \cite{Abra}, etc. However, most of these algorithms suffer from certain limitations in real-world applications, such as the loss of usual information, classification cost, excessive time, and adjustments, see \cite{Abra}.\\

In this paper, we address the problem of fraud detection in imbalanced data using the Poisson process; fraud is defined as a rare event occurring at a random time and involving significant financial losses. In this context, the fraud times are defined as the jump times of the Poisson process with intensity that describes the instantaneous rate of fraud.  Unlike machine learning methods, we do not look inside the subtle correlations in the data; instead, we assume that an exogenous rate or intensity must be determined. Instead of asking why the fraud is committed, the fraud rate is calibrated using market data. A lot of research has been done on the application of the Poisson process to financial risks, see \cite{AD}, \cite{JT}, \cite{DS}, etc. For calibration purposes, we assume that intensity is a deterministic function of time that takes into account the homogeneous and inhomogeneous Poisson process. Three main inputs are needed to estimate the intensity: the deterministic form of the intensity function, the arrival times of the frauds and the labels. \\

The main contributions of the paper are:

\begin{enumerate}
\item  Even though the intensity based approach is used in many fields, such as credit risk models, we are among the first to apply this approach to fraud detection
\item The Poisson process is addressed to rare events and it requires few inputs for the estimation of the intensity; so, the risk of over-fitting and computational cost would be reduced.
\item The approach combined with the machine learning algorithms can conduct to a sophisticated technique for detecting frauds.
\end{enumerate}

The rest of the paper is organized as follows.  Section II defines the mathematical concepts of Poisson process; the homogeneous and the Inhomogeneous Poisson process are reviewed. The estimation of the intensity and the prediction of fraud events are discussed. In the section III, the model is applied to financial datasets and the results are presented. The dataset was provided by NetGuardians \footnote{https://netguardians.ch}, a swiss company which develops solutions for banks to proactively prevent fraud.

\section{Mathematical concepts of Poisson Process}
\subsection{Fraud Event}
\label{arb}

Consider a financial institution such as a bank, an insurance company, a trading company, etc. and information about its clients. We are interested in the occurrence of fraud in client transactions for such an institution. The fraud event is then defined as a rare event occurring at a random time and resulting in significant financial losses for the client and the financial institution.
Whatever the definition used for a fraudulent event, let us note the fraud time by $\tau$ which corresponds to $[0, \infty] $ value of random variable on the filtered probability space $(\Omega, \mathcal {F}, \mathbb {F}, \mathbb {P}) $. 
$\Omega$ denotes the possible states of the world, $\mathcal{F}$ is the $\sigma$-algebra, $\mathbb{F}=(\mathcal{F}_{t})_{t\geq 0}$ is the filtration with $\mathcal{F}_{t}$ contains all information up to time $t$ and $\mathcal{F}_{T}=\mathcal{F}$. $\mathbb{P}$ is the probability measure describing the likelihood of certain events. The only mathematical structure assumed for $\tau$ is that it should be a stopping time, that is a random variable $\tau: \Omega \rightarrow \mathbb{R_{+}\cup \{\infty\}}$, such that $\{\tau\leq t\} \in \mathcal{F}_{t}$ for $t\geq 0$. Intuitively, one can determine whether or not the fraud time occurs before a certain deterministic time by observing the past up to time $t$, which is encoded in the filtration ($\mathcal{F}_{t}$).\\

Now consider a sequence $(\tau_{n})_{n \geq 0}$ of fraud times and let $N=\lbrace N(t); t\geq 0 \rbrace$ be a counting process given by 
\begin{eqnarray}
N(t)=\sum_{n\geqslant0} 1_{\{\tau_{n}\leqslant t\}}
\end{eqnarray}
In other words, $N(t)$ counts the number of fraud events between $0$ and $t$. $N$ has the following properties: $1. N(t)\geq 0$; $2. N(t)$ is an integer; 3. For $s\leq t$, $N(s)\leq N(t)$. The last property implies that $N$ is a submartingale since $E\left( N(t) \vert \mathcal{F}_{s}\right) \geq N(s)$. Because of the last property, the Doob-Meyer theorem guarantees the existence of an increasing predictable process A called compensator starting at 0 such that $M = N-A$ is a martingale. The compensator $A$ is uniquely defined up and governs the distribution of $N$. We assume that the compensator $A$ is absolutely continuous w.r.t. Lebesgue measure such that there
is a non-negative, integrable and predictable intensity process $\lambda$ that satisfies 
\begin{eqnarray}
A(t)=\int^{t}_{0}\lambda(s)ds
\end{eqnarray}
The process $\lambda$ represents the conditionally expected number of events per unit of time
in the sense that, at any time $t$, the $\mathcal{F}_{t-}$ conditional probability of an event between $t$ and $t+h$ is approximatively $\lambda(t) h$ for small $h$, where $\mathcal{F}_{t-}$ contains all information just before time $t$. In fact, because $N$ has the predictable intensity process $\lambda$, $dN(t)-\lambda(t) h$ is a martingale increment, and heuristically we thus have 

\begin{eqnarray}
E\left( dN(t)-\lambda(t) h \vert \mathcal{F}_{t-} \right) =0
\end{eqnarray}

Since $\lambda$ is predictable, $\lambda(t) \in \mathcal{F}_{t-}$ so we can move $\lambda(t) h$ outside the expectation and obtain 

\begin{eqnarray}
E\left( dN(t)\vert \mathcal{F}_{t-} \right) =\lambda(t) h
\end{eqnarray}

and 

\begin{eqnarray}
E\left( N(t+h)-N(t) \vert \mathcal{F}_{t-} \right) =\lambda(t) h
\end{eqnarray}

For more details see \cite{Reiss} and \cite{Thoma}\\

In the rest of this paper, we will focus on the counting process with a deterministic intensity that gives rise to homogeneous and inhomogeneous Poisson process. In this context, the likelihood of fraud events will be derived and implemented. 

\subsection{Homogeneous Poisson Process}
\subsubsection{Overview}
\label{arb}
The Homogeneous Poisson Process (HPP) is a fundamental stochastic process which is simple, easy to
understand and possesses desirable mathematical and theoretical properties making it easy to
handle. It can be easily extended to more complicated and realistic situations \cite{KI}. Let $N=(N(t))_{t\geq 0}$ be the counting process defined above i.e. for each $t >0$ which counts the number of fraud events that happen between time $0$ and time $t$. In order to have an overview of the Poisson process, let's consider three definitions of the Poisson process that are equivalent to each other. For the proof see \cite{RO} and \cite{Draz}.
\theoremstyle{definition}
\begin{definition}{}
$N$ is an HPP with constant intensity $\lambda \geq 0$ if:
\begin{enumerate} 
\item $N(0)=0$;
\item The process has stationary and independent increments;
\item For small $h$, $P(N(t+h)-N(t)=1)=\lambda h +o(h)$
\item $P(N(t+h)-N(t)\geq 2)=o(h)$
\end{enumerate}
\end{definition}

\begin{definition}{}
$N$ is an HPP with constant intensity $\lambda \geq 0$ if:
\begin{enumerate} 
\item $N(0)=0$;
\item The process has stationary and independent increments;
\item For $0 \leq s < t, N(t)-N(s)$ is Poisson distributed with parameter $\lambda(t-s)$.
\end{enumerate}
That is,
\begin{eqnarray}
P(N(t)-N(s)=k)=\frac{e^{-\lambda(t-s)}(\lambda(t-s))^{k}}{k!}
\end{eqnarray}
For any interval for size $t$, $\lambda t$ is the expected number of frauds in that interval.
\end{definition}

\begin{definition}{}
$N$ is HPP with constant intensity $\lambda \geq 0$ if the waiting times between successive events, or arrivals follow an exponential distribution of parameter $\lambda$.\\

\noindent This definition made the Poisson process unique among renewal process by the memoryless of the Exponential distribution. 

\end{definition}
\subsubsection{Estimation of the intensity $\lambda$}
\label{lam}
The simple and trivial way for estimating the constant intensity $\lambda$ is to use the above third definition of the HPP related to the Exponential distribution of the waiting times. \\

Let $N=(N(t))_{t\geq 0}$ an homogeneous Poisson process of parameter $\lambda$ and $(\tau_{n})_{n \geq 0}$ a sequence of fraud times. We define $S_{n}=\tau_{n}-\tau_{n-1}$, the waiting times between the event $n-1$ and the event $n$ with $S_{1}=\tau_{1}$. Because $(S_{t})_{t\geq 0}$ follows the exponential distribution of parameter $\lambda$,

\begin{eqnarray}
E(S_{n})=\frac{1}{\lambda}
\end{eqnarray}
Using the moment method, the estimator $\overline{\lambda}$ of $\lambda$ is given by
\begin{eqnarray}
\overline{\lambda}=\frac{1}{\overline{S}}
\label{int0}
\end{eqnarray}
which is also the Maximum Likelihood Estimator (MLE) of $\lambda$

\subsection{Non-Homogeneous Poisson Process}
\subsubsection{Overview}
\label{arb}
Non-Homogeneous Poisson Process (NHPP) means that the intensity $\lambda(t)$ is deterministic function. Thus, the distribution of the number of events between two particular points on the timeline is no longer a function depending on the difference between these points, as in the case of a Homogeneous Poisson Process (HPP). Here it is a function of the starting-point and the en point of the time interval and is not necessarily stationary. Let's start with the definition of the NHPP given in \cite{RO} 
\begin{definition}{}
The counting process $N=(N(t))_{t\geq 0}$ is said to be a NHPP with intensity function $\lambda(t)$, 
$t\geq 0 $, if it satisfies,
\begin{enumerate} 
\item $N(0)=0$;
\item $N$ has independent increments;
\item for small $h$, $P(N(t+h)-N(t)=1)=\lambda(t) h +o(h)$
\item $P(N(t+h)-N(t)\geq 2)=o(h)$
\end{enumerate}
The function $\lambda(t)$ is sometimes called the instantaneous arrival rate of the NHPP.
\end{definition}{}
A consequence of the above definition is that $N(t)-N(s)$ follows Poisson distribution of parameter $\int^{t}_{s}\lambda(u)du$. That is,
 \begin{eqnarray}
P(N(t)-N(s)=k)=\frac{e^{-\int^{t}_{s}\lambda(u)du}(\int^{t}_{s}\lambda(u)du)^{k}}{k!}
\end{eqnarray}
We can explore the relationship between the average number of events occurring up to the time $t$ and the intensity function $\lambda(t)$ of the corresponding NHPP:
\begin{eqnarray}
E(N(t))=\int^{t}_{0}\lambda(s)ds=A(t)
\end{eqnarray}
As described above, the compensator $A(t)$ is a non-decreasing right-continuous function and is referred here as the expectation function of the NHPP.\\
In addition, the expected number of events between times $t$ and $t+s$ is expressed as
\begin{eqnarray}
E(N(t+s)-N(t))=\int^{t+s}_{t}\lambda(u)du=A(t+s)-A(t)
\end{eqnarray}
According to \cite{CoxLewis},  we can examine the distribution function of the time to the next event in NHPP by 
\begin{eqnarray}
P(\mbox{1 or more events occurred in} \,  (t,t+s])=1-e^{-\int^{t+s}_{t} \lambda(u)du}=1-e^{-(A(t+s)-A(t))}
\label{sep}
\end{eqnarray}
Let $t_{s}=t+s$, the probability density function of the time to the next event, which can be obtained by deriving the expression in (\ref{sep}) with respect to $t_{s}$
\begin{eqnarray}
\frac{d}{t_{s}}  P(\mbox{1 or more events occurred in} \,  (t,t_{s}])=\lambda(t_{s})e^{-(A(t_{s})-A(t))}
\label{sep1}
\end{eqnarray}
As we will see later, this expression (\ref{sep1}) is very useful in estimating of the intensity.
\subsubsection{Estimation of the intensity $\lambda(t)$}
Suppose we have data from a non-homogeneous Poisson process $N=(N(t))_{t\geq 0}$ and we are looking for the intensity function that caused it. The first step is to define the form of the intensity $\lambda(t)$; we limit ourselves to the case of parametric intensity. In the second step, given the probability density function defined in (\ref{sep1}) we can use the principle of Maximum Likelihood Estimate (MLE) to find the intensity parameter $\lambda$ maximizing the likelihood that a fraud will occur. The procedure is the following:\\

Suppose the $n$ events occur at $\tau_{1}<\tau_{2}< ...<\tau_{n}$ in the interval $(0,T]$. Since the $n$ events are independent and using (\ref{sep1}), the desired joint probability density takes the form
\begin{eqnarray*}
\lambda(\tau_{1})e^{-(A(\tau_{1})-A(0))}\cdot\lambda(\tau_{2})e^{-(A(\tau_{2})-A(\tau{1}))}\cdot...\cdot\lambda(\tau_{n})e^{-(A(\tau_{n})-A(\tau_{n-1}))}\\
\cdot P(N(T)-N(\tau_{n})=0)
\end{eqnarray*}
where \\

$P(N(T)-N(\tau_{n})=0)$ is the  probability of no event occurs in the interval $(\tau_{n},T]$. It is calculated as follows
\begin{eqnarray*}
P(N(T)-N(\tau_{n})=0)=e^{-(A(T)-A(\tau_{n}))}
\end{eqnarray*}
The likelihood of getting $\tau={\tau_{1},\tau_{2}, ...,\tau_{n}}$ is then
\begin{eqnarray*}
L(\lambda;\tau={\tau_{1},\tau_{2}, ...,\tau_{n}})=e^{-A(T)}\prod_{i=1}^{n}\lambda(\tau_{i})
\end{eqnarray*}
The Log-Likelihood is :
\begin{eqnarray}
l(\lambda;\tau={\tau_{1},\tau_{2}, ...,\tau_{n}})&=&-A(T)+\sum_{i=1}^{n}\mbox{log}(\lambda(\tau_{i}))\\
&=&-\int^{T}_{0}\lambda(s)ds +\sum_{i=1}^{n}\mbox{log}(\lambda(\tau_{i}))
\label{log}
\end{eqnarray}
The intensity estimate consists in finding the parameters of the intensity $\lambda(t)$ maximizing the Log-likelihood function defined in (\ref{log}). This estimated intensity is then used to predict the fraud event on the next transaction ($T+1$) based on the information available up to the time of the transaction $T$.

\subsection{Prediction of Fraud Event}

Consider the filtration $\mathcal{F}_{T}$ that contains the information about the fraud events up to time $T$. Suppose a new transaction is in progress at time $T_{\delta}$ ($T_{\delta}>T$) and we would like to know if this transaction is fraudulent or not. \\

\begin{proposition}
The probability that a fraud occurs at time $T_{\delta}$ is given by 
\begin{eqnarray}
P(\mbox{a fraud occurs at } \, T_{\delta})&=&1-e^{-(A(T_{\delta})-A(T))}
\label{fou}
\end{eqnarray}
where\\

$A(T)=\int^{T}_{0}\lambda(s)ds$
\end{proposition}

\begin{proof}
Following (\ref{sep})
\begin{eqnarray*}
P(\mbox{a fraud occurs at} \, T_{\delta})&=&1-P(\mbox{a fraud does not occur at} \,  T_{\delta})\\
&=&1-P(N(T_{\delta})-N(T)=0 \vert \mathcal{F}_{T})\\
&=&1-e^{-\int^{T_{\delta}}_{T} \lambda(u)du}\\
&=&1-e^{-(A(T_{\delta})-A(T))}
\end{eqnarray*}
\end{proof}
In the special case of homogeneous Poisson process, that is for constant $\lambda$
\begin{eqnarray}
P(\mbox{a fraud occurs at} \,  T_{\delta})&=&1-e^{-\lambda(T_{\delta}-T)}
\label{fouc}
\end{eqnarray}
We observe that in the case of homogeneous Poisson process, the probability of fraud is a function of parameter $\lambda$ and the elapsed time ($T_{\delta}-T$) between the two transactions. For the Inhomogeneous Poisson, it is actually a function of the difference between the compensator $A(T_{\delta})$ and $A(T)$.\\

Following (\ref{fou}): as $(T_{\delta}-T)\rightarrow \infty$, $(A(T_{\delta})-A(T))\rightarrow \infty$ and then the \\$P(\mbox{the fraud occurs at} \, T_{\delta})\rightarrow1$. On another side, as $(T_{\delta}-T)\rightarrow 0$, $(A(T_{\delta})-A(T))\rightarrow 0$ and then the
$P(\mbox{the fraud occurs at} \, T_{\delta})\rightarrow 0$.\\

Therefore, when the time between two transactions is large, it is very likely that the model generates a fraud alert. On the other hand, when two transactions are close, the model will not generate a fraud alert. This consequence could reduce the predictive power of the model when there is a succession of fraud events in record time.

\section{Application to Financial Dataset}

\subsection{Choice of deterministic intensity functions}
To apply the Poisson process to the dataset, the shape of the intensity function must be defined. Three classes of intensity functions are proposed. For each class of function $\lambda(t)$, we set the conditions for $\lambda(t)\geq 0$.
\begin{enumerate}
\item $\lambda(t)=\lambda$: this is the case of Homogeneous Poisson process and $\lambda$ must be greater than 0. $\lambda$ is estimated following \S \hspace{0.01mm}  \ref{lam}.
\item  $\lambda(t)=a+bt$: the intensity is assumed to be a linear function of time. To ensure $\lambda(t) \geq 0$ for $0 \leq t \leq T$, we impose the conditions  
\begin{eqnarray}
\left\{
\begin{array}{l}
  a\geq 0 \\
  b+\dfrac{a}{T} \geq 0
\end{array}
\right.
\label{cond}
\end{eqnarray}

\begin{proof}
We want $\lambda(t)=a+bt \geq 0$ for $0 \leq t \leq T$.\\

\underline{If $t=0$}:\\
$\lambda(t)=a,\,\,\,         \lambda(t)\geq 0 \Rightarrow a\geq 0$\\

\underline{If $0 < t \leq T$}:\\
$a+bt \geq 0 \Leftrightarrow b \geq -\dfrac{a}{t}$\\
We also know that $-\dfrac{a}{t} \leq -\dfrac{a}{T}$ since $a\geq 0$. In order to have $b \geq -\dfrac{a}{t}$, it is sufficient that $b \geq -\dfrac{a}{T} \Leftrightarrow b+\dfrac{a}{T} \geq 0$. So, the conditions are $a\geq 0$ and $b+\dfrac{a}{T} \geq 0$. 

\end{proof}
If $T \rightarrow \infty $, we obtain the trivial condition

\begin{eqnarray*}
\left\{
\begin{array}{l}
  a\geq 0 \\
  b\geq 0
\end{array}
\right.
\end{eqnarray*}

Therefore, when we consider a short period to estimate the intensity parameters, the feasible region of (\ref{cond}) expands to find the optimal solution. The figure \ref{fig:fig22} shows the feasible regions for different values of $T$. For sake of readability, $a\leqslant 10$ and $b \leqslant 100$. We observe that when $T$ becomes larger, the feasible region is reduced to the trivial region.

\item $\lambda(t)=a+bt+ct^{2}$: the intensity is a quadratic function as a function of time. To ensure $\lambda(t) \geq 0$ for $0 \leq t \leq T$, we impose the conditions  

\begin{eqnarray}
\left\{
\begin{array}{l}
  a\geq 0 \\
 c\geq 0 \\
 b+\dfrac{a}{T} \geq 0
\end{array}
\right.
\label{cond1}
\end{eqnarray}


The proof is similar to the above. Also, when when $T \rightarrow \infty $, (\ref{cond1}) is reduced to 
\begin{eqnarray*}
\left\{
\begin{array}{l}
  a\geq 0 \\
  b\geq 0\\
  c \geq 0
\end{array}
\right.
\end{eqnarray*}
\end{enumerate}

The conditions (\ref{cond}) and (\ref{cond1}) are the constraints of the optimization problem in (\ref{log}) for the Inhomogeneous Poisson process.

\renewcommand{\thesubfigure}{\roman{subfigure}}
\begin{figure}[H]
\centering
\begin{subfigure}[b]{1.0\textwidth}
\centering
\includegraphics[width=0.8\textwidth]{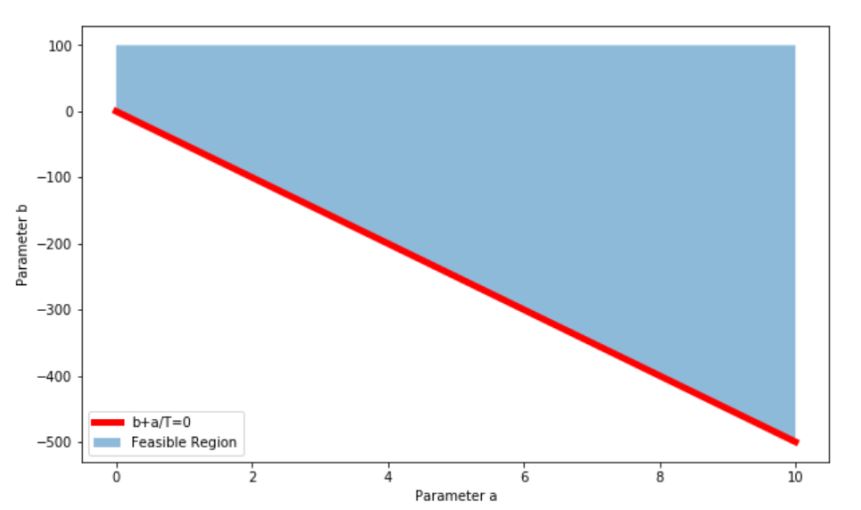}
\caption{}
\label{subfig-1:fig13} 
\end{subfigure}
\medskip
\begin{subfigure}[b]{1.0\textwidth}
\centering
\includegraphics[width=0.8\textwidth]{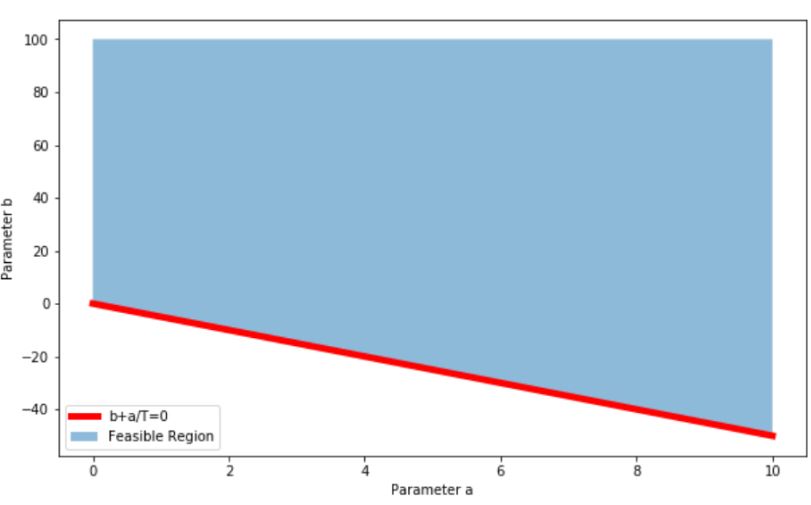}
\caption{}
\label{subfig-2:fig14}
\end{subfigure}
\label{fig:fig21}
\end{figure}

\begin{figure}[H]
\ContinuedFloat
\captionsetup{list=no}
\centering
\begin{subfigure}[b]{1.0\textwidth}
\centering
\includegraphics[width=0.8\textwidth]{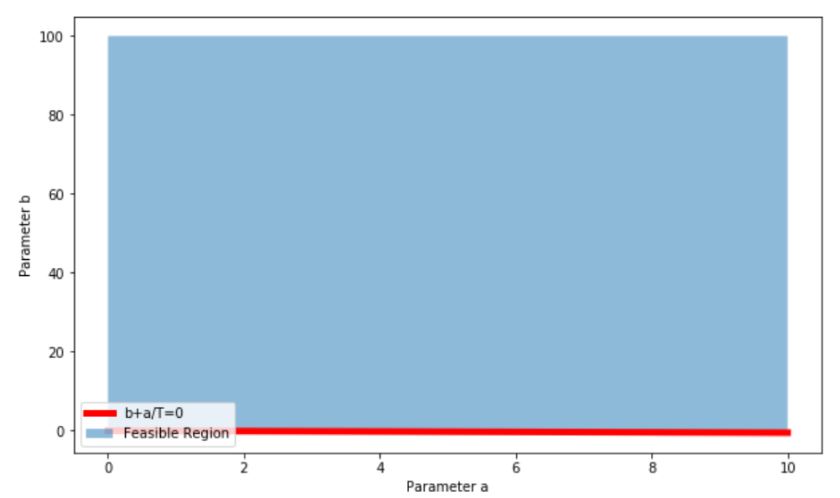}
\caption{}
\label{subfig-1:fig15}
\end{subfigure}

\caption{(i) Example of feasible region for $\lambda(t)=a+bt \geq 0$ when $T=0.02$: solution is in (\ref{cond}). The region is shown for $a\leqslant 10$ and $b \leqslant 100$. (ii) as for (i) but $T=0.2$. (iii) as for (i) but $T=20$.}
\label{fig:fig22}
\end{figure}

\subsection{Data}
The datasets provided by NetGuardians consist of two years of transactions for clients of a financial institution. It covers the period from 09-2015 to 09-2017 and includes a total of $18,139,078$ transactions made by $124,177$ clients. For confidentiality reasons, the name of the financial institution will not be mentioned. The dataset includes a total of $49$ features such as transaction dates, transactions amounts, transaction senders IDs, transaction recipients account numbers, banking countries, etc.. To be able to train a Poisson process algorithm, labelled data with examples of fraud are needed. All transactions in the dataset are labeled as fraudulent or not. Since the ground truth is not available, the labeling is based on the following simple pattern: transactions for which banks receiving money are outside Switzerland are considered fraudulent.
With the labelling method only $55,226$ clients have fraudulent transactions. To train the Poisson process, three features are required: client ID, timestamp and the label. Timestamps and labels are trained for each client to estimate the intensity of the fraud that will be used to predict  fraudulent event. \\

The proportion of fraud corresponding to the number of fraudulent transactions in relation to the total number of transactions is calculated for each client. According to the labeling method, some clients may have a 100\% fraud proportion. This concerns clients for whom the recipient institutions are all located outside Switzerland. To be realistic, we remove these clients from our analysis. In addition, clients that do not contain any fraud events in the complete dataset are deleted because the hours of fraud events are unknown and their intensity can not be estimated. In addition, these datasets contain only one class and, in this context, no measure of classification performance such as ROC-AUC is defined.\\

The figure \ref{fig:fig23}  shows the distribution and the Boxplot of fraud proportions.  We notice that the cleaned dataset is generally unbalanced because most clients have a low proportion of frauds. The Boxplot shows a skewed right data with the presence of larger outliers. With the value of the median, $50\%$ of the clients have a fraud proportion less than $9\%$. 

\renewcommand{\thesubfigure}{\roman{subfigure}}
\begin{figure}[H]
\centering
\begin{subfigure}[b]{1.0\textwidth}
\centering
\includegraphics[width=0.9\textwidth]{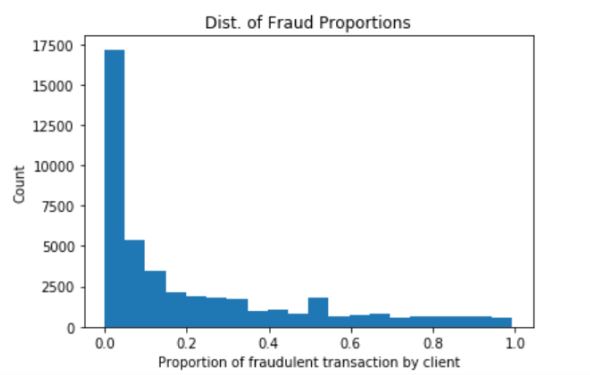}
\caption{}
\label{subfig-1:fig1} 
\end{subfigure}
\medskip
\begin{subfigure}[b]{1.0\textwidth}
\centering
\includegraphics[width=0.8\textwidth]{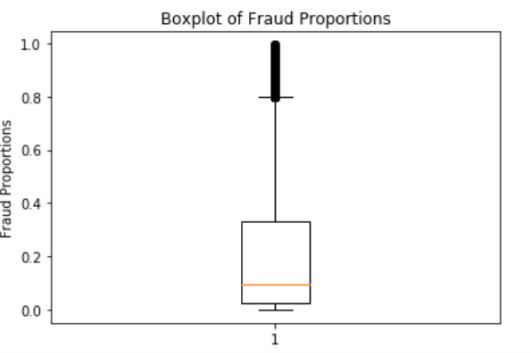}
\caption{}
\label{subfig-2:fig100}
\end{subfigure}
\caption{(i) Histogram of Fraud proportions in the full dataset. (ii) Boxplot of Fraud proportions in the full dataset. The clients with no fraud events and the clients with $100\%$ of fraud proportion are removed from this full dataset}
\label{fig:fig23}
\end{figure}

However, it is important to mention that the labelling method is relatively simple and that the above histogram is not representative of the true distribution of fraud because, in practice, the majority of fraud proportions are less than $1\%$. To study our analysis in an imbalanced dataset framework,
we propose to focus on the clients with less than $20\%$ frauds. Next, we divide this dataset into four subsets containing different fraud profiles. The first subset includes clients fraud rate less than $1\%$, the second subset concerns clients with a proportion between $1\%$ and $5\%$, the third subset is for clients whose fraud proportion is between $5\%$ and $10\%$ and the last one for clients whose fraud proportion is between $10\%$ and $20\%$. Figure \ref{fig2} shows the Boxplot for each group. The four datasets are roughly symmetric with no outliers. Obviously, the greater variability in the group 4 and the  smaller variability in the group 1 are well observed. 

\begin{figure}[H]
\begin{center}
\includegraphics[width=4.5in]{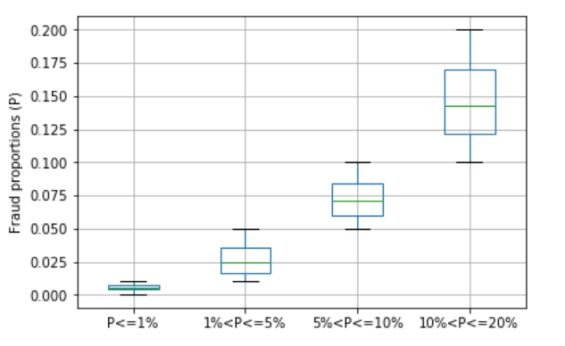}
\end{center}
\caption{Boxplots for the four subsets. Clients with the fraud proportion $P\leq 20\%$ are grouped in four subsets containing different fraud profiles. }
\label{fig2}
\end{figure}

In each subset, we randomly select  $500$ clients and we train and test the Poisson models on the transactions for each client. The training set represents the first $80\%$ of transactions for which intensity parameters are estimated. The test set represents the last $20\%$ and the fraud events are predicted with the estimated parameters. In addition, to take into account the time-varying intensity parameters, the prediction in the test set is also performed by rolling windows.\\

From a practical point of view, when there is no fraud in the training set, it is difficult to estimate the fraud intensity because the fraud event times are not available; see equations (\ref{int0}) and (\ref{log}). Two solutions are possible:

\begin{enumerate}
\item Remove the clients for whom there was no fraud occurrence in the training set; The consequence is that we could lose more information. 
\item Make the assumption that the intensity i.e. the occurrence rate of fraud $\lambda=0$ as there is no fraud events in the training set. In this context the fraud prediction probability is zero; see Proposition $1$.  
\end{enumerate}

We conduct our analysis with the last one that is the intensity $\lambda=0$ when we train a dataset with no fraud information. The main reason is that we expect to keep most of client profiles in our analysis. As we will see later, under this assumption the dynamic models perform worse than the static models.  
To compare the various Poisson models, we define a baseline model (benchmark) based on a naive approach. The naive approach is to calculate the proportion of fraud in the training set and use that probability to predict fraud in the test set.  Finally, predictive performance is summarized in each subset using two performance measures: ROC-AUC and Average Precision (AP) Score.  \\
\subsection{Results}
By adding the rolling windows approach to our study, we have a total of 6 models to compare. Let start by giving more explanations to the $6$ models: 
\begin{enumerate}
\item The first model is the homogeneous Poisson process ($\lambda(t)=\lambda$). The constant intensity $\lambda$ is estimated in the training set. By (\ref{fouc}), the estimated $\lambda$ is used for predicting the fraud event in the whole test set.   We note this model by \textbf{HomoStatic}.
\item The second model is the Homogeneous Poisson process unless the prediction is done by rolling windows. The window starts by the training set and it is used for the estimation of  the intensity; this estimated intensity is used to predict the fraud event on the next transaction in the test set. Then, the sliding window is shifted one step ahead on the next transaction. The intensity is estimated again in the second time window and it is used for the prediction of fraud on the next transaction. This procedure is repeated until the end of the test set. The goal of this methodology is to take account the time varying of the intensity. The model is denoted by \textbf{HomoDynamic}.
\item The third model is the non-homogeneous Poisson process with the intensity is a linear function of time ($\lambda(t)=a+bt$). Intensity paramters are estimated in the training set and are used for the fraud prediction in the whole test set. It is denoted \textbf{LinearStatic}.
\item The fourth model is the inhomogeneous linear intensity function unless the prediction is performed by rolling windows. The rolling windows procedure is the same as above. It is denoted \textbf{LinearDynamic}
\item The fifth model is the non-homogeneous Poisson process with the intensity being a quadratic function of time ($\lambda(t)=a+bt+ct^{2}$). The procedure is the same as in \textbf{LinearStatic}. We denote this model \textbf{QuadraticStatic} 
\item The last model is as \textbf{QuadraticStatic} unless we make a prediction by rolling windows. It is denoted by \textbf{QuadraticDynamic}
\end{enumerate}
In addition, we note by \textbf{NaiveStatic} the baseline model to estimate the probability of fraud in the training set and using the same probability for the prediction in the test set. The probabilities of prediction are therefore the same for all the transactions of the test set. This is equivalent to a random classifier because the model has no discrimination capability to distinguish genuine transactions from fraudulent transaction. \\

We are interested in the power of prediction of the different models. Thus, all the results presented below are based on the predicting probabilities and the labels in the test set. 
The tables \ref{table:1} to \ref{table:4} show the AUC (Area Under The curve)-ROC (Receiver Operating Characteristics) curves for the different models in each group. AUC-ROC is the measure of performance for the classification problem at various thresholds settings. ROC is a probability curve and AUC represents the degree or measure of separability.  It tells how much the model is able to distinguish between classes. Higher the AUC, better the model is. By analogy, higher the AUC, better the model is at distinguishing between genuine and fraudulent transactions. The tables show the mean, the standard deviation, the minimum and maximum for the AUCs calculated for 500 clients in each group.\\

We note that dynamic models (with rolling windows) are more volatile than static models (without rolling windows). All static models perform significantly better than the dynamic models. The LinearStatic model is the best one and has a mean AUC of $69\%$, $73\%$, $72\%$, $71\%$ in the group 1, group 2, group 3 and group 4 respectively. It is followed by the QuadraticStatic model. The baseline model (naive approach) is significantly worse than Poisson models with the exception of the QuadraticDynamic model in the group 1 where the mean AUC is $47\%$. However, the HomoDynamic model performs better than the other dynamic models. 
It is important to mention that in some cases Poisson do not predict frauds correctly, as AUCs are equal to $0$. It is often the case when  
the fraud information used in the training set to estimate the intensity is not sufficient for the prediction in the test set. Let us illustrate one common situation in our dataset where there is no fraud in the training set that conducts to AUC=0. Consider an example of dataset with $6$ training instances and $3$ test instances. The labels are:\\

\textbf{Training set:} $[0\,\,\, 0\,\,\, 0\,\,\, 0\,\,\, 0 \,\,\,0]$ \hspace{5mm} \textbf{Test set:} $[1 \,\,\,0\,\,\, 0]$\\

The labels $0$ indicate genuine transactions and labels $1$ indicate fraudulent transactions. There is no fraud events in the training set and from the above assumption $\lambda=0$. For all static models, the prediction probabilities in the test set are $0$ and therefore the AUC-ROC is equal to $0.5$. On the other hand, the dynamic models based on the sliding windows show an AUC-ROC equal to $0$. In fact, it is easy to show that using the sliding windows in the test set, the first predicting probability is $0$ and the next two ones are different to $0$. This conducts to an AUC-ROC equal $0$. \\

\begin{table}[h!]
\centering
\begin{tabular}{||c c c c c||} 
 \hline
 Models & Max & Mean & Min & Standard Deviation\\ [0.5ex] 
 \hline\hline
 HomoDynamic & 1 & 0.503217 & 0 & 0.341836\\ 
 HomoStatic & 1 & 0.674555 & 0.434599 & 0.227029\\
 LinearDynamic & 1 & 0.499829 & 0 & 0.339872\\
 LinearStatic & 1 & 0.684265 & 0.434599 & 0.235382\\
 QuadraticDynamic & 1 & 0.471564 & 0 & 0.313377\\ 
 QuadraticStatic & 1 & 0.676788 & 0.434599 & 0.231372\\   
 NaiveStatic & 0.5 & 0.50 & 0.50 & 0\\ [1ex] 
 \hline
\end{tabular}
\caption{AUC: Summary for statistics in the group 1 ($P\leq1\%$)}
\label{table:1}
\end{table}

\begin{table}[h!]
\centering
\begin{tabular}{||c c c c c||} 
 \hline
 Models & Max & Mean & Min & Standard Deviation\\ [0.5ex] 
 \hline\hline
 HomoDynamic & 1 & 0.658384 & 0 & 0.297473\\ 
 HomoStatic & 1 & 0.716914 & 0.048193 & 0.246125\\
 LinearDynamic & 1 & 0.639246 & 0 & 0.295391\\
 LinearStatic & 1 & 0.732957 & 0.048193 & 0.246664\\
 QuadraticDynamic & 1 & 0.612797 & 0 & 0.287221\\ 
 QuadraticStatic & 1 & 0.727479 & 0.048193 & 0.243740\\   
 NaiveStatic & 0.5 & 0.50 & 0.50 & 0\\ [1ex] 
 \hline
\end{tabular}
\caption{AUC: Summary for statistics in the group 2 ($1\%<P\leq5\%$)}
\label{table:2}
\end{table}

\begin{table}[h!]
\centering
\begin{tabular}{||c c c c c||} 
 \hline
 Models & Max & Mean & Min & Standard Deviation\\ [0.5ex] 
 \hline\hline
 HomoDynamic & 1 & 0.682609 & 0 & 0.273724\\ 
 HomoStatic & 1 & 0.709963 & 0 & 0.251212\\
 LinearDynamic & 1 & 0.670611 & 0 & 0.268545\\
 LinearStatic & 1 & 0.717206 & 0 & 0.244682\\
 QuadraticDynamic & 1 & 0.649789 & 0 & 0.262771\\ 
 QuadraticStatic & 1 & 0.714477 & 0 & 0.243915\\   
 NaiveStatic & 0.5 & 0.50 & 0.50 & 0\\ [1ex] 
 \hline
\end{tabular}
\caption{AUC: Summary for statistics in the group 3 ($5\%<P\leq10\%$)}
\label{table:3}
\end{table}

\begin{table}[h!]
\centering
\begin{tabular}{||c c c c c||} 
 \hline
 Models & Max & Mean & Min & Standard Deviation\\ [0.5ex] 
 \hline\hline
 HomoDynamic & 1 & 0.675263 & 0 & 0.254579\\ 
 HomoStatic & 1 & 0.695023 & 0 & 0.246416\\
 LinearDynamic & 1 & 0.655216 & 0 & 0.262901\\
 LinearStatic & 1 & 0.709309 & 0 & 0.241264\\
 QuadraticDynamic & 1 & 0.650595 & 0 & 0.261025\\ 
 QuadraticStatic & 1 & 0.708917 & 0 & 0.240598\\   
 NaiveStatic & 0.5 & 0.50 & 0.50 & 0\\ [1ex] 
 \hline
\end{tabular}
\caption{AUC: Summary for statistics in the group 4 ($10\%<P\leq20\%$)}
\label{table:4}
\end{table}

AUC-ROC can be a misleading measure for classification in imbalanced fraud dataset. One of the main reason is that it underestimates the false positive rate. In fact, since the number of legitimate transactions (negative examples) far exceeds the number of fraudulent transactions (positive examples), a significant variation in the number of false positives can lead to a slight change in the false positive rate. This can lead to erroneous conclusions.  
In this case the precision-recall analysis is more appropriate because these metrics do not take into account the number of legitimate transactions (negative examples) in their calculation. We focus on the Average Precision (AP) which is an estimate of the area under the precision-recall curve and their results are shown in the following tables \ref{table:5} to \ref{table:8}. All the Poisson models significantly outperform the naive approach and static approaches perform better than the dynamic approaches. LinearStatic model still remains the better one for all groups, following by the QuadraticStatic model. Also, the HomoDynamic model performs better than the other dynamic models. In conclusion, the AUC-ROC and AP analyses  showed that in all four groups the linearStatic model is the best; it is followed by the QuadraticStatic model and then by the HomoDynamic model. All the Poisson models outperform significantly the baseline approach.\\

We are also interested in the relative performance in term of prediction between the Poisson models and the baseline approach. The idea is to determine in which group the Poisson models perform best. AP scores are used for this analysis. The relative variations between the Mean Average-Precision (MAP) for the different Poisson models and the baseline model are calculated in the table \ref{table:9}. The table shows that the relative variation decreases when the fraud proportion of the group increases. So, the predicting power of the Poisson models increases with the degree of imbalanced dataset. Figure \ref{fig12} shows the relative performance for the different models in each group. We observe that the relative performance is better in the group $1$ and that the linearStatic model outperforms the other $5$ models.\\

During the analysis, we observe that dynamic approaches (Rolling Windows) are less efficient than the static approaches regardless the performance measures. That is, taking account the temporal variation of the intensity parameters by the rolling windows does not produce better results. Two mains reasons could explain this weak performance of dynamic models. First, as illustrated above, the assumption of $\lambda=0$ when we train a dataset with no fraud may conduct to this weak performance. Second, the window size is essential for the the forecast accuracy. In fact following \cite{Roll}, different window sizes may lead to different empirical results in practice and good results might be obtained simply by chance. To produce better results, one can vary the window size and select the optimal window size for better prediction. Another possibility is to consider a stochastic intensity model that incorporates the time varying of the parameters. This has to be conducted in a next research.\\

\begin{table}[h!]
\centering
\begin{tabular}{||c c c c c||} 
 \hline
 Models & Max & Mean & Min & Standard Deviation\\ [0.5ex] 
 \hline\hline
 HomoDynamic & 1 & 0.274416 & 0.004132 & 0.400234\\ 
 HomoStatic & 1 & 0.332728 & 0.004132 & 0.430593\\
 LinearDynamic & 1 & 0.202763 & 0.004132 & 0.325295\\
 LinearStatic & 1 & 0.390334 & 0.004132 & 0.456534\\
 QuadraticDynamic & 1 & 0.135732 & 0.004132 & 0.263333\\ 
 QuadraticStatic & 1 & 0.354866 & 0.004132 & 0.439122\\   
 NaiveStatic & 0.05 & 0.022027 & 0.001511 & 0.011081\\ [1ex] 
 \hline
\end{tabular}
\caption{AP: Summary for statistics in the group 1 ($P\leq1\%$)}
\label{table:5}
\end{table}
\hspace{2em}

\begin{table}[h!]
\centering
\begin{tabular}{||c c c c c||} 
 \hline
 Models & Max & Mean & Min & Standard Deviation\\ [0.5ex] 
 \hline\hline
 HomoDynamic & 1 & 0.451901 & 0.005307 & 0.354692\\ 
 HomoStatic & 1 & 0.511768 & 0.005307 & 0.369290\\
 LinearDynamic & 1 & 0.380657 & 0.005216 & 0.331117\\
 LinearStatic & 1 & 0.566920 & 0.005358 & 0.361079\\
 QuadraticDynamic & 1 & 0.279384 & 0.005320 & 0.266282\\ 
 QuadraticStatic & 1 & 0.548226 & 0.005358 & 0.362918\\   
 NaiveStatic & 0.25 & 0.062533 & 0.007576 & 0.047267\\ [1ex] 
 \hline
\end{tabular}
\caption{AP: Summary for statistics in the group 2 ($1\%<P\leq5\%$)}
\label{table:6}
\end{table}
\hspace{2em}

\begin{table}[h!]
\centering
\begin{tabular}{||c c c c c||} 
 \hline
 Models & Max & Mean & Min & Standard Deviation\\ [0.5ex] 
 \hline\hline
 HomoDynamic & 1 & 0.540694 & 0.029412 & 0.318448\\ 
 HomoStatic & 1 & 0.578589 & 0.029412 & 0.312256\\
 LinearDynamic & 1 & 0.497483 & 0.029412 & 0.302288\\
 LinearStatic & 1 & 0.598028 & 0.029412 & 0.303077\\
 QuadraticDynamic & 1 & 0.440209 & 0.029412 & 0.288972\\ 
 QuadraticStatic & 1 & 0.589859 & 0.029412 & 0.299586\\   
 NaiveStatic & 0.5 & 0.141468 & 0.010638 & 0.113300\\ [1ex] 
 \hline
\end{tabular}
\caption{AP: Summary for statistics in the group 3 ($5\%<P\leq10\%$)}
\label{table:7}
\end{table}
\hspace{2em}

\begin{table}[h!]
\centering
\begin{tabular}{||c c c c c||} 
 \hline
 Models & Max & Mean & Min & Standard Deviation\\ [0.5ex] 
 \hline\hline
 HomoDynamic & 1 & 0.560771 & 0.040000 & 0.286504\\ 
 HomoStatic & 1 & 0.599116 & 0.040000 & 0.283211\\
 LinearDynamic & 1 & 0.540127 & 0.040000 & 0.286291\\
 LinearStatic & 1 & 0.623278 & 0.040000 & 0.271896\\
 QuadraticDynamic & 1 & 0.517028 & 0.040000 & 0.286320\\ 
 QuadraticStatic & 1 & 0.622054 & 0.040000 & 0.270560\\   
 NaiveStatic & 0.8 & 0.203014 & 0.018519 & 0.123488\\ [1ex] 
 \hline
\end{tabular}
\caption{AP: Summary for statistics in the group 4 ($10\%<P\leq20\%$)}
\label{table:8}
\end{table}
\hspace{2em}

\begin{table}[h!]
\centering
\begin{tabular}{||c c c c c||} 
 \hline
 Models & $P \leq 1\%$ & $1\%<P\leq 5\%$ & $5\%<P\leq 10\%$ & $10\%<P\leq 20\%$\\ [0.5ex] 
 \hline\hline
 HomoDynamic & 11.458411 & 6.226555 & 2.822037 & 1.762236\\ 
 HomoStatic & 14.105765 & 7.183912 & 3.089905 & 1.951114\\
 LinearDynamic & 8.205382 & 5.087254 & 2.516586 & 1.660544\\
 LinearStatic & 16.721054 & 8.065881 & 3.227313 & 2.070128\\
 QuadraticDynamic & 5.162206 & 3.467756 & 2.111730 & 1.546766\\ 
 QuadraticStatic & 15.110806 & 7.766931 & 3.169568 & 2.064102\\ [1ex]  
  \hline
\end{tabular}
\caption{Relative Variations of MAP between the Poisson Models and the Baseline model in the four groups}
\label{table:9}
\end{table}

\begin{figure}[H]
\begin{center}
\includegraphics[width=6in]{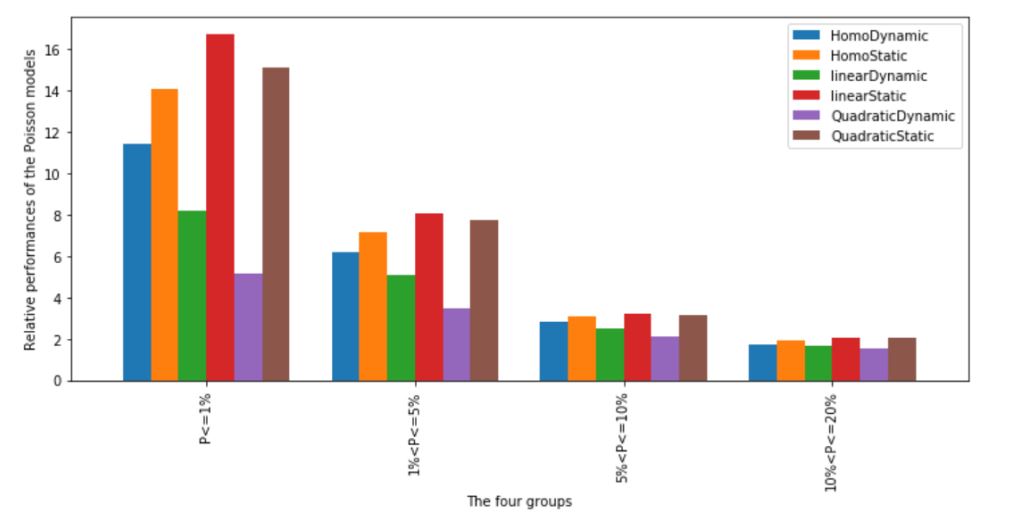}
\end{center}
\caption{Relative performances between the different models and the baseline approach.  These performances are plotted in each group showing in which group  the Poisson models perform the best }
\label{fig12}
\end{figure}

\section{Conclusion}
The Poisson process is applied to detect fraud in an imbalanced dataset. The case of homogeneous and non-homogeneous Poisson processes is investigated. For non-homogeneous Poisson process, the linear and quadratic functions are considered. We have shown how to estimate the intensity and to predict fraud events. Our methodology is applied to financial datasets.\\

For each Poisson model studied, we consider the static and the dynamic approach. Unlike the static approach, the dynamic one takes into account the temporal variation of intensity parameters and works with rolling windows. All models are compared to a baseline model of fraud prediction using the proportion of frauds obtained in the training set. We found that all Poisson models outperform the baseline and that static approaches perform better than the dynamic ones. The static linear model remains the better for all groups followed by the static quadratic model and then by the homogeneous Poisson model. The study also showed a better predicting power of the Poisson models in the case of more imbalanced dataset.\\

One of the main problems of this study is the training of the Poisson process in a set with no fraud events. In this context, it is difficult to estimate the intensity parameters because we have no fraud event times. In this study, it is assumed that the  intensity is zero. But as indicated above this assumption could conduct to a poorer performance of the model. \\

Another problem is the dynamic of the intensity function. It is assumed here that the fraud rate is constant or deterministic i.e. function of time. In fact, fraud is a rare event that can happen at any times; so it must be stochastic, a random variable at any time. These issues will be addressed in future research by detecting fraud using a stochastic intensity model combined with deep learning algorithms.
\newpage
\bibliographystyle{apalike}
\bibliography{toto}

\end{document}